\documentclass[11pt]{article}
\usepackage{fullpage}
\usepackage{amsmath,amssymb,xspace,amsfonts,amsthm}
\usepackage{bm}

\newtheorem{Thm}{Theorem}
\newtheorem{Lem}[Thm]{Lemma}
\newtheorem{Cor}[Thm]{Corollary}
\newtheorem{Prop}[Thm]{Proposition}

\newtheorem{Def}{Definition}
\newtheorem{Fact}[Thm]{Fact}

\newcommand\h{\mathcal{H}}
\newcommand\ket[1]{| #1 \rangle}
\newcommand\bra[1]{\langle #1 |}

\newcommand\mbR{\mbox{$\mathbb{R}$}}
\newcommand\mbC{\mbox{$\mathbb{C}$}}
\newcommand\mcX{\mathcal{X}}

\newcommand\size{\mbox{\tt {size}}\xspace}
\newcommand\rank{\mbox{\tt {rank}}\xspace}
\newcommand\srank{\mbox{\tt {S-rank}}\xspace}
\newcommand\prank[1]{\mbox{\tt {rank}$^{(#1)}_{\tt psd}$}\xspace}
\newcommand\pranktwo{\mbox{\tt {rank}$_{\tt psd}$}\xspace}

\newcommand\apprank[1]{\mbox{\tt {rank}$_{{\tt psd},#1}$}\xspace}
\newcommand\defeq{\stackrel{\mathrm{\scriptsize def}}{=}}
\newcommand\tr{\mbox{\tt {tr}}\xspace}
\newcommand\qip[2]{\langle #1 | #2 \rangle}
\newcommand\F{\mbox{\sf {F}}\xspace}

\newcommand\alice{\mbox{\sf Alice}\xspace}
\newcommand\bob{\mbox{\sf Bob}\xspace}
\newcommand\charlie{\mbox{\sf Charlie}\xspace}
\newcommand\qcomm{\mbox{\sf {QComm}}\xspace}
\newcommand\rcomm{\mbox{\sf {RComm}}\xspace}
\newcommand\qcorr{\mbox{\sf {QCorr}}\xspace}

\newcommand\m{\mbox{\sf {M}}\xspace}

\newcommand {\st} {\textit{s.t.}\xspace}
\newcommand {\ie} {\textit{i.e.}\xspace}

\begin{document}

\title{Multipartite Quantum Correlation and Communication Complexities}
\author{Rahul Jain\thanks{Department of Computer Science and Centre for Quantum Technologies, National University of Singapore, Singapore. Email: rahul@comp.nus.edu.sg}  \\
\and
Zhaohui Wei\thanks{School of Physics and Mathematical Sciences, Nanyang Technological University and Centre for Quantum Technologies, Singapore. Email: weizhaohui@gmail.com} \\
\and
Penghui Yao\thanks{Centre for Quantum Technologies, National University of Singapore, Singapore. Email: phyao1985@gmail.com} \\
\and Shengyu Zhang\thanks{Department of Computer Science and Engineering, and The Institute of Theoretical Computer Science and Communications, The Chinese University of Hong Kong. Email: syzhang@cse.cuhk.edu.hk}}

\date{}

\maketitle

\begin{abstract}
\small%\baselineskip=9pt
The concepts of quantum correlation complexity and quantum communication complexity were recently proposed to quantify the minimum amount of resources needed in generating bipartite classical or quantum states in the single-shot setting. % by Zhang ({\em Proc. 3rd Innov Theor Comput., pp. 39-59, 2012}).
The former is the minimum size of the initially shared state
$\sigma$ on which local operations by the two parties (without
communication) can generate the target state $\rho$, and the latter
is the minimum amount of communication needed when initially sharing
nothing. In this paper, we generalize these two concepts to
multipartite cases, for both exact and approximate state generation.
Our results are summarized as follows.
\begin{enumerate}
  \item For multipartite pure states, the correlation complexity can be completely characterized by local ranks of sybsystems.

    \item We extend the notion of PSD-rank of matrices to that of tensors, and use it to bound the quantum correlation complexity for generating multipartite classical distributions.

    \item For generating multipartite mixed quantum states, communication complexity is not always equal to correlation complexity (as opposed to bipartite case). But they differ by at most a factor of 2.
  Generating a multipartite mixed quantum state has the same communication complexity as generating its optimal purification. But for correlation complexity of these two tasks can be different (though still related by less than a factor of 2).

  \item To generate a bipartite classical distribution $P(x,y)$ approximately, the quantum communication complexity is completely characterized by the approximate PSD-rank of $P$. The quantum correlation complexity of approximately generating multipartite pure states is bounded by approximate local ranks. %This result is obtained based on the fact that the cost to approximate a bipartite quantum state equals that to approximate its exact purifications, which also implies a result for a general bipartite quantum state.

\end{enumerate}
\end{abstract}

\pagestyle{plain} \pagenumbering{arabic}

\section{Introduction}

Shared randomness and quantum entanglement among parties located at
different places are important resources for various
distributed information processing tasks. How to generate these shared resources has been one of the most important issues, and recently much attention has been paid to the minimum amount of shared correlation and communication needed to generate bipartite classical and quantum states in one-shot setting
\cite{ASTS+03,HJMR09,FMP+12,Zha12,JSWZ13}. In particular, in \cite{JSWZ13} the worst-case costs of several single-shot bipartite schemes to generate correlations and quantum entanglement have been characterized. %, and the specific problems they considered and their results are as follows.
The setting is as follows. Suppose that two parties, \alice and
\bob, need to generate correlated random variables $X$ and $Y$, with
\alice outputting $X$ and \bob outputting $Y$, such that $(X,Y)$ is
distributed according to a target distribution $P$. If $P$ is not a
product distribution, \alice and \bob could generate $P$ by sharing
an initial seed distribution $(X',Y')$, \alice owning $X'$ and \bob
owning $Y'$, and then each performing local operations on their own
part. The minimal size of this seed correlation $(X',Y')$ is defined
as {\em randomized correlation complexity} \cite{Zha12}, denoted
$R(P)$, where the {\em size} of a bipartite distribution is defined
as the half of the total number of bits. It has been known that
$R(P)$ is fully characterized as $\lceil \log_2 \rank_+(P)\rceil$
\cite{Zha12}, where $\rank_+(P)$ is the nonnegative rank of matrix
$P$\footnote{A bipartite distribution $P$ is also natural a matrix
$[P(x,y)]_{x,y}$, and we thus use $P$ for both the distribution and
the matrix.}, a measure in linear algebra with numerous applications
in combinatorial optimization \cite{Yan88}, nondeterministic
communication complexity \cite{LL90}, algebraic complexity theory
\cite{Nis91}, and many other fields \cite{CP05}. The problem becomes
even more interesting when quantum operations are allowed: \alice
and \bob share a seed quantum state $\sigma$ and perform local
quantum operations to generate a distributed classical distribution
$P$. In this case, the minimal size of the seed quantum state
$\sigma$ is defined as {\em quantum correlation complexity}, denoted
$\qcorr(P)$, where the {\em size} of a bipartite quantum state is
the half of the total number of qubits. One of the main results of
\cite{JSWZ13} is that $\qcorr(P)$ could be completely characterized
as $\lceil \log_2 \pranktwo(P)\rceil$, where $\pranktwo(P)$ is the
PSD-rank of $P$, a concept recently proposed by Fiorini {\em et al.}
in studies of the minimum size of extended formulations of
optimization problems such as TSP \cite{FMP+12}. Since $\rank_+(P)$
could be much larger than $\pranktwo(P)$, this implies a potentially
huge advantage of using quantum operations, over classical
counterparts, to generate classical distributions.

More generally, the target state can be a quantum state $\rho$, and \cite{JSWZ13} gave a complete
characterization for the minimum size of the seed state to generate $\rho$. In particular, if $\rho$ is a pure state
$\ket{\psi}\bra{\psi}$ and an $\epsilon$-approximation is allowed for generating $\ket{\psi}\bra{\psi}$, then the correlation complexity is completely characterized by the
$(1-\epsilon)^2$-cutoff point of the Schmidt coefficients of $\ket{\psi}$, closing a possibly exponential
gap left in \cite{ASTS+03}.

The above discussion assumes that \alice and \bob perform local operations on a shared state. Actually, \alice and \bob could replace the shared states discussed above by communication. In this case,
the minimal amount of communication in classical and quantum protocols for generating target classical distribution $P$
are defined as the {\em randomized} and {\em quantum communication complexities}, denoted by
$\rcomm(P)$ and $\qcomm(P)$, respectively \cite{Zha12}; one can similarly define $\qcomm(\rho)$ for generating quantum states $\rho$. We have introduced correlation complexity and communication complexity. An interesting fact is that when only two parties are involved, these two measures are always the same, and this is true for both classical and quantum settings \cite{Zha12}.

%For convenience, in this case we let $\Q(P)=\qcorr(P)=\qcomm(P)$,
%and here $P$ could also be replaced by a quantum state $\rho$.

Capturing the minimum cost to generate target states, the concepts
of correlation complexity and communication complexity are
fundamental parameters of the shared states as a resource. In
particular, when the target state is quantum, the resource is
entanglement, arguably the most important shared resource in almost
all quantum information processing tasks. While bipartite
entanglement is well understood, multipartite entanglement has been
elusive on many levels, and considerable efforts have been made to
study it from various angles. In this paper, we extend the study of
correlation and communication complexity of generating a classical
correlation and quantum entanglement to multipartite cases. Our
results are summarized next, and we hope that they can shed light on
multipartite entanglement from another fundamental perspective.

\subsection{Multipartite quantum correlation complexity}\label{intro:3}

For multipartite cases, it turns out that quantum correlation
complexity and quantum communication complexity are not equivalent
any more, thus we have to deal with them separately. We first
consider quantum correlation complexity of generating a $k$-partite state.
\begin{Def}\label{def:qcc}
Suppose $k$ parties, $A_1$, $A_2$, ..., $A_k$, share a seed state
$\sigma$, and they aim to generate a target state $\rho$ by each perform some operation on her own part of $\sigma$. Then the quantum correlation complexity of $\rho$, denoted $\qcorr(\rho)$, is the minimal size of $\sigma$ such that
local quantum operations on $\sigma$ can generate $\rho$. Here the size of $\sigma$ is defined as $\sum_{i=1}^kn_i$, where $n_i$ is the number of qubits of $\sigma$ held by $A_i$.
\end{Def}

Let us first consider the case of $\rho$ being a pure state. For a bipartite pure state
$\ket{\psi}$, Schmidt decompositions help us to characterize
$\qcorr(\ket{\psi})$ and $\qcomm(\ket{\psi})$ perfectly, but
multipartite pure states do not have Schmidt decompositions in general. It turns out that the quantum correlation complexity is the sum of the ``marginal complexity''.
\begin{Def}
    Suppose $\ket{\psi}$ is a pure state in $\h_{1}\otimes \cdots \otimes\h_{k}$, and $\rho_{j}$ is the reduced density matrices of $\ket{\psi}$ in $\h_{j}$. Define the marginal complexity of $\ket{\psi}$ as
\begin{align*}
   \m(\ket{\psi}) = \sum_{j=1}^k \big\lceil \log_2 \rank(\rho_{j})\big\rceil.
\end{align*}
\end{Def}

\newtheorem*{Thm-qccforpure}{Theorem~\ref{thm:qccforpure}}
\begin{Thm}\label{thm:qccforpure}
Suppose $\ket{\psi}$ is a pure state in $\h_{1}\otimes \cdots \otimes\h_{k}$, and $\rho_{j}$ is the reduced density matrices of $\ket{\psi}$ in $\h_{j}$. Then
\begin{align*}
   \qcorr(\ket{\psi}) = \m(\ket{\psi}).
\end{align*}
\end{Thm}

\medskip For a mixed quantum state $\rho$, however, the correlation complexity is less clear. It was
mentioned that in the bipartite case, $\qcorr(\rho)$ is exactly the minimal $\qcorr(\ket{\psi})$ over all purifications $\ket{\psi}$ of $\rho$ \cite{JSWZ13}. This turns out to be not the case any more in multipartite setting.

\newtheorem*{Thm-qcc_intro}{Theorem~\ref{thm:qcc_intro}}
\begin{Thm}\label{thm:qcc_intro}
Assume that $\rho$ is a quantum state in $\bigotimes_{i=1}^k \h_{i}$. Then we have
\begin{align*}
    \qcorr(\rho)\leq r(\rho) \leq \Big(2-\frac{2}{k}\Big)\qcorr(\rho),
\end{align*}
where $r(\rho)$ is the minimum $\qcorr(\ket{\psi})$ over all purifications $\ket{\psi}$ of $\rho$.
\end{Thm}
We will also show that both inequalities in the above theorem are
tight, thereby implying that $\qcorr(\rho)$ is indeed different from
$\min\{\qcorr(\ket{\psi}): \ket{\psi} \text{ purifies } \rho\}$.

\medskip
While in some sense pure quantum states contain the most
``quantumness'' in terms of superposition, the other extreme is
mixture of classical states, \ie, classical distributions. In the
bipartite case, the quantum correlation complexity of generating
distribution $P = [P(x,y)]_{x,y}$ is exactly $\lceil \log_2
\pranktwo(P)$. We will show an analogous result in multipartite
cases. To this end, we need to first generalize the notion of
PSD-rank from matrices to tensors. Similar to the bipartite case, a
$k$-partite probability distribution $P =
[P(x_1,x_2,...,x_k)]_{x_1,x_2,...,x_k}$ can also be viewed as a
tensor of dimension $k$.
\begin{Def}\label{def:psd}
For an entry-wise nonnegative tensor $P =
[P(x_1,...,x_k)]_{x_1,...,x_k}$ of dimension $k$, its PSD-rank
$\prank{k}(P)$ is the minimum $r$ \st there are $r\times r$ PSD
matrices $C_{x_1}^{(1)}, ..., C_{x_k}^{(k)}\succeq 0$ with
$P(x_1,...,x_k) = \sum_{i,j=1}^r C_{x_1}^{(1)}(i,j) \cdots
C_{x_k}^{(k)}(i,j)$.
\end{Def}
With this definition, we can bound the quantum correlation
complexity of $P$ in terms of its PSD-rank.
\newtheorem*{Thm-qccP}{Theorem~\ref{thm:qccP}}
\begin{Thm}\label{thm:qccP}
Suppose $P = [P(x_1,...,x_k)]_{x_1,...,x_k}$ is a probability distribution on $\mcX_1\times \cdots \times \mcX_k$. Then \begin{align*}
    \frac{k}{2k-2}\big\lceil\log_2\prank{k}(P)\big\rceil \leq \qcorr(P)\leq k\big\lceil\log_2\prank{k}(P)\big\rceil.
\end{align*}
\end{Thm}

\subsection{Multipartite quantum communication complexity}\label{intro:4}
As earlier mentioned, one can also consider the setting in which the players share nothing at the beginning and communicate to generate some target state. The communication complexity is formally defined as follows.
\begin{Def}\label{def:qcomm}
Suppose $k$ parties, $A_1$, $A_2$, ..., $A_k$, initially share nothing and aim to jointly generate a quantum state
$\rho$ by communication. The quantum communication complexity of generating $\rho$, denoted $\qcomm(\rho)$, is the minimum number of qubits exchanged between these $k$ parties.
\end{Def}
The following theorem gives bounds for quantum communication complexity of pure states. Recall that $\m(\ket{\psi}) = \sum_{j=1}^k \lceil \log_2 \rank(\rho_i) \rceil $, where $\rho_i$ is $\ket{\psi}$ reduced to Player $i$'s space.

\newtheorem*{Thm-qcomm_pure}{Theorem~\ref{thm:qcomm_pure}}
\begin{Thm}\label{thm:qcomm_pure}
Suppose $\ket{\psi}$ is a $k$-partite pure state. Then
\begin{align*}
   \frac{1}{2}\m(\ket{\psi})\leq\qcomm(\ket{\psi}) \leq \frac{k-1}{k}\m(\ket{\psi}).
\end{align*}
\end{Thm}

Next we turn to general multipartite quantum mixed states. Different than quantum correlation complexity,
the quantum communication complexity $\qcomm(\rho)$ is always equal to the minimum $\qcomm(\ket{\psi})$ over
purifications $\ket{\psi}$ of $\rho$.

\newtheorem*{Thm-qcomm_topure}{Theorem~\ref{thm:qcomm_topure}}
\begin{Thm}\label{thm:qcomm_topure}
For any $k$-partite quantum state $\rho$,
\begin{align*}
    \qcomm(\rho) = \min \{\qcomm(\ket{\psi}): \ket{\psi} \text{ is a purification of }\rho\}.
\end{align*}
\end{Thm}

Combining the results in the above two subsections together, we get
the following relationship between $\qcorr(\rho)$ and $\qcomm(\rho)$
for a general multipartite quantum state $\rho$.
\newtheorem*{Thm-comm-vs-corr}{Corollary~\ref{Thm:comm-vs-corr}}
\begin{Cor}\label{Thm:comm-vs-corr}
For any $k$-partite quantum state $\rho$,
\begin{align*}
    \frac{k}{k-1}\qcomm(\rho)\leq\qcorr(\rho)\leq2\qcomm(\rho).
\end{align*}
\end{Cor}

\subsection{Approximate quantum correlation complexity}
In this section, we consider relaxing the task of state generation
by allowing approximation. After all, we usually generate the state
for some later information processing purpose, and thus if the
generated state $\rho'$ is very close to the target state $\rho$,
then the same precision can be preserved after whatever further
operations, global or local.

\subsubsection{Bipartite}\label{intro:1}
When a good approximation instead of the exact generation is satisfactory, the minimum size of the shared seed state can be smaller than that for the exact generation. In \cite{JSWZ13}, a natural definition for the approximate correlation complexity was given as follows.
\begin{Def}
Let $\rho$ be a bipartite quantum state in $\h_{A} \otimes \h_{B}$, and $\epsilon>0$. Define
\begin{equation*}\label{def:forappqcc}
\qcorr_\epsilon(\rho) \defeq \min\{ \qcorr(\rho') : \rho' \in \h_{A}
\otimes \h_{B}\nonumber \text{ and } \F(\rho, \rho') \geq 1 -
\epsilon\}
\end{equation*}
\end{Def}
Since for any bipartite state $\rho$, generating a mixed state is the same as generating its (optimal) purification \cite{Zha12,JSWZ13}
\[\qcorr(\rho)= \min\{\qcorr(\ket{\psi}): \ket{\psi} \text{ purifies } \rho\} = \min\{\lceil \log_2 \srank(\ket{\psi})\rceil: \ket{\psi} \text{ purifies } \rho\},\]
%and that
%\[\qcorr(\rho)= \min\{\qcorr(\ket{\psi}): \ket{\psi} \text{ purifies } \rho\} = \min\{\lceil \log_2 \srank(\ket{\psi})\rceil: \ket{\psi} \text{ purifies } \rho\}.\]
it is also natural to given another definition by putting the approximation on the purification instead of the original target state. Let
\[\qcorr_\epsilon'(\rho) = \min\{\qcorr(\ket{\psi'}): \ket{\psi} \text{ purifies } \rho, \F(\ket{\psi},\ket{\psi'}) \ge 1-\epsilon\}.\]
As we will show, these two definitions are equivalent, \ie,
$\qcorr_\epsilon(\rho) = \qcorr_\epsilon'(\rho)$. Note that the
second definition is easier to analyze, since the approximate
correlation complexity for pure states are well understood
\cite{JSWZ13}: For any $\ket{\psi} =
\sum_{x,y}A(x,y)\ket{x}\ket{y}$, let matrix $A = [A(x,y)]$, then
\[\min\{\qcorr(\ket{\psi'}): \F(\ket{\psi},\ket{\psi'}) \ge
1-\epsilon\} = \rank_{2\epsilon-\epsilon^2}(A),\] where the
approximate rank $\rank_\delta(A)$ of a matrix $A$ is the smallest
number $r$ \st the summation of the largest $r$ singular values
squared is at least $1-\delta$.

%Therefore, the approximate quantum correlation complexity of $\rho$ is determined by the minimal Schmidt rank of its ``approximate purification". Here an approximate purification of $\rho$ means a pure state with the reduced density matrix on A and B close to $\rho$. However, usually the pattern of approximating a mixed state is very complicated, thus the approximate purifications seems not easy to analyze. In the following theorem, however, we will see that this difficulty can be avoided, and approximate correlation complexity can be characterized completely by approximating the {\em exact} purifications of $\rho$, and this problem has been well-solved \cite{JSWZ13}.

%\begin{Lem}\label{thm:qcc1_intro}
%Assume that $\rho$ is a quantum state in $\h_A \otimes \h_{B}$. Let
%\begin{align*}
%    \qcorr'_\epsilon(\rho) = \min_{\h_{A_1}, \h_{B_1}}\{\big\lceil\log_2\srank_\epsilon(\ket{\varphi})\big\rceil: \ket{\varphi} \text{ is a pure state} \\
%        \text{ in } \h_{A_1}\otimes \h_A\otimes \h_B \otimes \h_{B_1},  \rho= \tr_{\h_{A_1} \otimes \h_{B_1}}
%        \ket{\varphi}\bra{\varphi}\},
%\end{align*}
%then $\qcorr_\epsilon(\rho)=\qcorr'_\epsilon(\rho)$.
%\end{Lem}

Based on this result, we could get the following
characterization of $\qcorr_\epsilon(\rho)$ for the special case of
classical $\rho$, namely when $\rho$ is a classical distribution $P$. We first define approximate PSD-rank and
approximate correlation complexity by classical states as follows.
\begin{Def}\label{def:approximatePSD}
$P = [p(x,y)]_{x,y}$ is a bipartite probability distribution, its
$\epsilon$-approximate PSD-rank is
\begin{equation}\label{def:appprank}
\apprank{\epsilon}(P)=\min\{\pranktwo(P'): F(P,P')\geq1-\epsilon\}.
\end{equation}
where $P'$ is another probability distribution on the same sample space of $P$.
\end{Def}
\newcommand\cla{cla}
\begin{Def}
    For a bipartite classical distribution $P = [P(x,y)]_{x,y}$, its {\em $\epsilon$-approximate quantum correlation complexity by classical state} is $\qcorr_\epsilon^{\cla}(P) = \min\{\qcorr(P'): F(P,P') \geq 1-\epsilon\}$, where $P'$ is another probability distribution on the same sample space of $P$.
\end{Def}
The following theorem says that the most efficient approximate
generation of a classical state can always be achieved by another
classical state. Moreover, the approximate correlation complexity of a classical
state could be completely characterized by the approximate PSD-rank.

\newtheorem*{Thm-aqcorr_P}{Theorem~\ref{thm:aqcorr_P}}
\begin{Thm}\label{thm:aqcorr_P}
For any classical state $P = [P(x,y)]_{x,y}$, \[\qcorr_\epsilon(P) =
\qcorr_\epsilon^{\cla}(P) = \lceil \log_2 \apprank{\epsilon}(P) \rceil.\]
\end{Thm}

Finally, for the general case of an arbitrary quantum state $\rho$, we give the following characterization of
$\qcorr_{\epsilon}(\rho)$.

\newtheorem*{Thm-aQCorrgeneral}{Theorem~\ref{thm:aQCorrgeneral}}
\begin{Thm}\label{thm:aQCorrgeneral}
Let $\sigma$ be an arbitrary quantum state in $\h_A \otimes \h_B$,
and $0<\epsilon<1$. Then $\qcorr_{\epsilon}(\sigma) = \lceil \log_2 r
\rceil $, where $r$ is the minimum integer \st there exist a collection of matrices, $\{A_x\}$'s and $\{B_y\}$'s of the same column number $l\ge r$, satisfying the following conditions.

\begin{enumerate}
    \item The matrices relate to $\sigma$ by the following equation.
    \begin{align}\label{eq:requireforgeneral}
        \sigma=\sum_{x,x'; y,y'} \ket{x}\bra{x'} \otimes \ket{y}\bra{y'} \cdot \tr \Big( (A_{x'}^\dag A_x)^T (B_{y'}^\dag B_y) \Big)\Big).
    \end{align}

    \item Denoting the $i$-th column of any matrix $M$ by $\ket{M(i)}$, then
    \begin{equation}\label{eq:requireforgeneral2}
        \sum_x\qip{A_x(i)}{A_x(j)}=\sum_y\qip{B_y(i)}{B_y(j)}=0,
    \end{equation}

    \item
    \begin{align}\label{eq:requireforgeneral3}
        \sum_{i=1}^r\Big(\sum_x\qip{A_x(i)}{A_x(i)}\Big)\Big(\sum_y\qip{B_y(i)}{B_y(i)}\Big)\geq1-\epsilon,
    \end{align}
\end{enumerate}
\end{Thm}

\subsubsection{Multipartite}\label{intro:5}

%Without Schmidt decompositions, it will be challenging to characterize the approximate version of correlation complexity for multipartite pure states.
As in \cite{JSWZ13}, it is natural to consider two different approximations to a pure target state, one to approximate by a mixed state, and the other to approximate by a pure state.

\begin{Def}
Let $\epsilon>0$. Let $\ket{\psi}$ be a $k$-partite quantum pure state in
$\h_{1} \otimes \h_{2} \otimes...\otimes \h_{k}$. Define
\begin{equation*}\label{def:forappmulqcc}
\qcorr_\epsilon(\ket{\psi}) \defeq \min\{ \qcorr(\rho') : \rho'
\text{ is in } \h_{1} \otimes \h_{2} \otimes...\otimes
\h_{k}\nonumber \text{ and } \F(\ket{\psi}\bra{\psi}, \rho') \geq 1
- \epsilon\}
\end{equation*}
and
\begin{equation*}
\qcorr_\epsilon^{pure}(\ket{\psi}) \defeq \min\{ \qcorr(\ket{\phi}) :
\ket{\phi} \in \h_{1} \otimes \h_{2} \otimes...\otimes
\h_{k}\nonumber \text{ and } \F(\ket{\psi}\bra{\psi}, \ket{\phi}\bra{\phi}) \geq 1
- \epsilon\}.
\end{equation*}
\end{Def}
We can see that $\qcorr_\epsilon(\rho)$ and
$\qcorr_\epsilon^{pure}(\rho)$ are the complexities of approximating
$\rho$ by mixed and pure states respectively.

For a $k$-partite pure state $\ket{\psi}$ in $\h_1\otimes\h_2\otimes\cdots\otimes \h_k$, let $\rho_i$ be the reduced density matrix of $\ket{\psi}$ in $\h_i$, and $r_i = \rank(\rho_i)$. Denote the $\epsilon$-approximate Schmidt rank
of $\ket{\psi}$ with respect to the separation $(A_i,A_{-i})$ (here $A_{-i} = A_1...A_{i-1}A_{i+1}...A_k$) as $r_i^{(\epsilon)}$, i.e., $r_i^{(\epsilon)}=\srank^{(A_i,A_{-i})}_{\epsilon}(\ket{\psi})$. Then we have

\newtheorem*{Thm-aQccforPure}{Theorem~\ref{thm:aQccforPure_intro}}
\begin{Thm}\label{thm:aQccforPure_intro}
Let $\ket{\psi}\in \bigotimes_{i=1}^k \h_{i}$ be a $k$-partite state, $\epsilon>0$, and
%\begin{align*}
    $\m_{\epsilon}(\ket{\psi}) = \sum_{i=1}^k \big\lceil\log_2 r_{i}^{(\epsilon)}\big\rceil$.
%\end{align*}
Then
\begin{align*}
\m_{\epsilon}(\ket{\psi})\leq\qcorr_{\epsilon}^{pure}(\ket{\psi})\leq
\m_{\epsilon/k}(\ket{\psi}).
\end{align*}
\end{Thm}

Finally, we consider the relationship between $\qcorr_{\epsilon}(\ket{\psi})$ and $\qcorr^{pure}_{\epsilon}(\ket{\psi})$.
\newtheorem*{Thm-PureandMixed}{Theorem~\ref{thm:PureandMixed}}
\begin{Thm}\label{thm:PureandMixed}
Let $\ket{\psi}\in \h_{A_1}\otimes\cdots\otimes \h_{A_k}$ be a pure state and $\epsilon >0$. Then
\begin{align*}
\frac{k}{2k-2}\qcorr^{pure}_{k\epsilon}(\ket{\psi})\leq\qcorr_{\epsilon}(\ket{\psi})\leq
\qcorr^{pure}_{\epsilon}(\ket{\psi}).
\end{align*}
\end{Thm}

\section{Preliminaries}\label{sec:pre}
In this paper we consider multipartite systems. If a system has $k$ parties, we usually use $A_1$, ..., $A_k$ to denote them. Their spaces are $\h_1, ..., \h_k$, respectively. For notational convenience, we use $A_{-i}$ for $A_1...A_{i-1}A_{i+1}...A_k$, and use subscript $-i$ in other symbols (such as $\h_{-i}$) for a similar meaning.
\begin{trivlist}
\item{\bf Matrix theory}. For a natural number $n$ we let $[n]$ represent the set $\{1,2, \ldots, n\}$. We sometimes write $A=[A(x,y)]$ to mean that $A$ is a matrix with the $(x,y)$-th entry being $A(x,y)$. An operator $A$ is said to be {\em Hermitian} if $A^\dag = A$. A Hermitian operator $A$ is said to be \emph{positive semi-definite} (PSD) if all its eigenvalues are non-negative. %We will use the following fact.
%\begin{Fact}
%\label{fact:psd}
For any vectors $\ket{v_1}, \ldots, \ket{v_r}$ in $\mbC^n$, the $r
\times r$ matrix $M$ defined by $M(i,j) \defeq \qip{v_i}{v_j}$ is
positive semi-definite. The following definition of PSD-rank of a
matrix was proposed in \cite{FMP+12}.
\begin{Def}
    For a matrix $P\in \mbR_+^{n\times m}$, its PSD-rank, denoted $\pranktwo(P)$, is the minimum number $r$ such that there are PSD matrices $C_x,D_y\in \mbC^{r\times r}$ with $\tr(C_x D_y) = P(x,y)$, $\forall x \in [n], y\in [m]$.
\end{Def}
One can see that this corresponds to the special case of $k=2$ in Definition \ref{def:psd}. When $k=2$, we drop the superscript $(2)$ in Definition \ref{def:psd}, thus making it consistent with the above definition of PSD-rank of matrices.

\item{\bf Quantum information}.
A quantum state $\rho$ in Hilbert space $\h$, denoted $\rho \in \h$,
is a trace one positive semi-definite operator acting on $\h$. A
quantum state $\rho$ is called \emph{pure} if it is rank one, namely
$\rho = \ket{\psi}\bra{\psi}$ for some vector $\ket{\psi}$ of unit
$\ell_2$ norm; in this case, we often identify $\rho$ with
$\ket{\psi}$. For quantum states $\rho$ and $\sigma$, their fidelity
is defined as $\F(\rho,\sigma) \defeq
\tr(\sqrt{\sigma^{1/2}\rho\sigma^{1/2}})$. For $\rho, \ket{\psi} \in
\h$, we have $\F(\rho,\ket{\psi}\bra{\psi}) = \sqrt{\bra{\psi} \rho
\ket{\psi}}$. We define norm of $\ket{\psi}$ as $\|\ket{\psi}\|
\defeq \sqrt{\qip{\psi}{\psi}}$. For a quantum state $\rho \in  \h_A
\otimes \h_B$, we let $\tr_{\h_B} \rho$ represent the partial trace
of $\rho$ in $\h_A$ after tracing out $\h_B$. Let $\rho \in \h_A$
and $\ket{\phi} \in \h_A \otimes \h_B$ be such that $\tr_{\h_B}
\ket{\phi}\bra{\phi}  = \rho$, then we call $\ket{\phi}$ a {\em
purification} of $\rho$.
\begin{Def}
For a pure state $\ket{\psi} \in \h_A \otimes \h_B$, its {\em
Schmidt decomposition} is defined as
\begin{align*}
\ket{\psi} = \sum_{i=1}^r \sqrt{p_i} \cdot \ket{v_i} \otimes \ket{w_i},
\end{align*}
where $\ket{v_i}$'s are orthonormal states in $\h_A$, $\ket{w_i}$'s are orthonormal states in $\h_B$, and $p$ is a probability distribution.
\end{Def}
It is easily seen that $r$ is also equal to $\rank(\tr_{\h_A}
\ket{\psi}\bra{\psi} ) =\rank(\tr_{\h_B} \ket{\psi}\bra{\psi} ) $
and is therefore the same in all Schmidt decompositions of
$\ket{\psi}$.  This number is also referred to as the {\em Schmidt
rank} of $\ket{\psi}$ and denoted $\srank^{(A,B)}(\ket{\psi})$. The
superscript $(A,B)$ is to emphasize that the partition is between
$A$ and $B$. The next fact can be shown by considering Schmidt
decomposition of the pure states involved; see, for example,
Ex(2.81) of \cite{NC00}.
\begin{Fact}\label{fact:local}
Let $\ket{\psi}, \ket{\phi} \in \h_A \otimes \h_B$ be such that
$\tr_{\h_B} \ket{\phi}\bra{\phi} = \tr_{\h_B} \ket{\psi}\bra{\psi}
$. There exists a unitary operation $U$ on $\h_B$ such that
$(I_{\h_A} \otimes U) \ket{\psi} = \ket{\phi}$, where $I_{\h_A}$ is
the identity operator on $\h_A$.
\end{Fact}
We will also need another fundamental fact, shown by
Uhlmann~\cite{NC00}.
\begin{Fact}[Uhlmann, \cite{NC00}] \label{fact:uhlmann}
Let $\rho, \sigma \in \h_A$. Let $\ket{\psi} \in \h_A \otimes \h_B$
be a purification of $\rho$ and $\dim(\h_A) \leq \dim(\h_B)$. There
exists a purification $\ket{\phi} \in \h_A \otimes \h_B$ of $\sigma$
such that $\F(\rho, \sigma) = |\qip{\phi}{\psi}|$.
\end{Fact}
The approximate version of Schmidt decomposition that will be
utilized in the present paper is as follows, which is called {\em
approximate Schmidt rank}.

\begin{Def}
\label{def:appsrank} Let $\epsilon>0$. Let $\ket{\psi}$ be a pure
state in $\h_A \otimes \h_B$. Define
%%%%%%%%% double column %%%%%%%%%%
%\begin{align}
%\srank^{(A,B)}_\epsilon(\ket{\psi}) \defeq \min\{ \srank^{(A,B)}(\ket{\phi}) : \ket{\phi} \in \h_A \otimes \h_B \\
%\nonumber \text{ and } \F(\ket{\psi}\bra{\psi},
%\ket{\phi}\bra{\phi}) \geq 1 - \epsilon\} .
%\end{align}
%%%%%%%%% double column %%%%%%%%%%
%%%%%%%%% single column %%%%%%%%%%
$$\srank^{(A,B)}_\epsilon(\ket{\psi}) \defeq \min\{ \srank^{(A,B)}(\ket{\phi}) : \ket{\phi} \in \h_A \otimes \h_B \text{ and } \F(\ket{\psi}\bra{\psi}, \ket{\phi}\bra{\phi}) \geq 1 - \epsilon\} .$$
%%%%%%%%% single column %%%%%%%%%%
\end{Def}

For multipartite pure states, there are no Schmidt decompositions in general. But a weaker statement holds.
\begin{Lem}\label{lem:qccforpure}
Suppose $\ket{\psi}$ is a pure state in
$\h_1\otimes\h_2\otimes\cdots\otimes\h_{k}$, and $\rho_i$ is the
reduced density matrix of $\ket{\psi}$ in $\h_i$. Denote $r_i =
\rank(\rho_i)$. If $\{\ket{\alpha_{ij}}: j\in [r_i]\}$ are the
eigenvectors of $\rho_i$ corresponding to nonzero eigenvalues, then
$\ket{\psi}$ can be expressed as
\begin{align*}
   \ket{\psi}=\sum_{j_1\in[r_1], ..., j_k\in [r_k]} a_{j_1...j_k}\ket{\alpha_{1j_1}}\otimes \cdots \otimes \ket{\alpha_{kj_k}},
\end{align*}
where $a_{j_1...j_k}$'s are complex coefficients.
\end{Lem}
\begin{proof}
    For each $i\in [k]$, one can extend the vectors $\ket{\psi_{i1}}$, ..., $\ket{\psi_{ir_i}}$ to orthogonal basis $\ket{\psi_{i1}}$, ..., $\ket{\psi_{iD_i}}$ of $\h_i$, where $D_i$ is the dimension of $\h_i$. One can then decompose $\ket{\psi}$ according to the basis $\ket{\psi_{ij}}: i\in [k], j\in [D_i]$. The statement just says that $\ket{\psi}$ does not have any component in $\ket{\psi_{ij}}$, $\forall i$, $\forall j>r_i$. This is true because if $\ket{\psi}$ has a nonzero component in $\ket{\psi_{ij}}$ for some $j>r_i$, then when we compute the reduced density matrix of $\ket{\psi}$ in $\h_i$, we get $\rho_i$ with a positive component in $\ket{\psi_{ij}}\bra{\psi_{ij}}$. Thus $\ket{\psi_{ij}}$ is a eigenvector of $\rho_i$ with a nonzero eigenvalue, contradictory to our assumption.
\end{proof}

\end{trivlist}

\section{Quantum Correlation Complexity of Multipartite States}

In this section, we prove the results in Subsection \ref{intro:3} on quantum correlation complexity of multipartite states.

\begin{Thm-qccforpure}[Restated]
Suppose $\ket{\psi}$ is a pure state in $\h_{1}\otimes \cdots \otimes\h_{k}$, and $\rho_{j}$ is the reduced density matrices of $\ket{\psi}$ in $\h_{j}$. Then
\begin{align*}
   \qcorr(\ket{\psi}) = \sum_{j=1}^k \big\lceil \log_2 \rank(\rho_{j})\big\rceil.
\end{align*}
\end{Thm-qccforpure}

\begin{proof}
Let $r_j = \rank(\rho_j)$. By Lemma \ref{lem:qccforpure}, suppose
that $\ket{\psi} = \sum_{i_j \le r_j} a_{i_1 ... i_k}
\ket{\lambda_{i_1}}\cdots \ket{\lambda_{i_k}}$, where
$\ket{\lambda_{i_j}}$ is the $j$-th eigenvector of $\rho_j$. Then
the players can generate $\ket{\psi}$ by local operations on the
seed state $\ket{\psi'} = \sum_{i_j \le r_j} a_{i_1 ... i_k}
\ket{i_1}\cdots \ket{i_k}$. Since this state takes $\sum_{j=1}^k
\big\lceil \log_2 \rank(\rho_{j})\big\rceil$ number of qubits, we
have shown that $\qcorr(\ket{\psi}) \leq \sum_{j=1}^k \big\lceil
\log_2 \rank(\rho_{j})\big\rceil$.

For the other direction, let us assume the $k$ players generate the target $\ket{\psi}$ by local operations on an initial
seed state $\sigma$, whose size is $\qcorr(\ket{\psi})$. First note that to generate a pure state, it is enough to have a pure state as the seed, since otherwise every pure state in the support of the mixed seed state can give the same target $\ket{\psi}$.

Now define the reduced density matrix of $\sigma$ in the system
$A_j$ as $\sigma_j$, and assume that its rank is $s_j$. Then the
size of $\sigma$ is at least $\sum_{j=1}^k \big\lceil \log_2
s_j\big\rceil$, where the $j$-th summand bounds the number of qubits
for the $j$-th player's part of $\sigma$. Since local operations do
not increase Schmidt rank, we know that $s_j \ge r_j$. Thus
\[\qcorr(\ket{\psi}) \ge \sum_{j=1}^k \big\lceil \log_2 s_j \big\rceil
\ge \sum_{j=1}^k \big\lceil \log_2 r_j \big\rceil = \sum_{j=1}^k
\big\lceil \log_2 \rank(\rho_{j})\big\rceil.\]
\end{proof}

% Note that when $k=2$, the above theorem coincides with the one for the bipartite case in \cite{JSWZ13}.

As we mentioned earlier, generating a bipartite mixed quantum state
$\rho$ has the same cost as generating some purification of $\rho$
\cite{Zha12}. In multipartite cases, however, this does not hold any
more. The next theorem compares the quantum correlation complexity
of generating a mixed state $\rho$ and that of generating a
purification.

\begin{Thm-qcc_intro}[Restated]
Assume that $\rho$ is a quantum state in $\bigotimes_{i=1}^k \h_{i}$. Then we have
\begin{align*}
    \qcorr(\rho)\leq r(\rho) \leq \Big(2-\frac{2}{k}\Big)\qcorr(\rho),
\end{align*}
where $r(\rho)$ is the minimum $\qcorr(\ket{\psi})$ over all purifications $\ket{\psi}$ of $\rho$.
\end{Thm-qcc_intro}
\begin{proof}
First, we have $\qcorr(\rho)\leq\qcorr(\ket{\psi})$ for any purification $\ket{\psi}$ of $\rho$, thus $\qcorr(\rho) \leq
r(\rho)$.

Now for the other direction, suppose that a minimal seed state for
generating $\rho$ is $\sigma$ with $size(\sigma) = \qcorr(\rho)$.
Let $\sigma_i$ be the reduced density matrix of $\sigma$ in
$\h_{A_i}$, and suppose that $n_i$ is the number of qubits of
$\sigma_i$, so $\qcorr(\rho)=\sum_{i=1}^k n_i$. Without loss of
generality, assume that $n_1\leq \cdots \leq n_k$. Take any
purification $\ket{\theta}$ of $\sigma$ in $\h_{A_1} \otimes \cdots
\otimes \h_{A_{k-1}} \otimes \h_{A_k} \otimes \h_{A_k'}$, where
$A_k'$ is the ancillary system introduced by $A_k$. In each player's
part, the local operation can be assumed to be attaching some extra
system, performing a unitary operation, and then tracing out part of
system. Now if all players do not trace out any part of their
systems, and act on initial state $\ket{\theta}$ instead of
$\sigma$, then the same protocol results in a pure state
$\ket{\psi}$, which is a purification of $\rho$. In this way,
$\qcorr(\ket{\psi})\leq\qcorr(\ket{\theta})$.

According to Theorem \ref{thm:qccforpure}, we have
$\qcorr(\ket{\theta}) = \sum_{i=1}^k \lceil\log_2 r_i\rceil$, where $r_i$ is the dimension of $\sigma_i$ for $i\le k-1$, and $r_{k}$ is the dimension of $\tr_{\h_1\otimes\cdots\otimes\h_{k-1}} \ket{\theta}\bra{\theta}$. Note that
\[r_i \le 2^{n_i}, \ \forall i\le k-1,\ \text{ and } r_k \le 2^{n_1+\cdots+n_{k-1}},\]
where the last inequality uses the fact that $\ket{\theta}$ is a pure state. Thus, it follows that
\begin{align*}
    \qcorr(\ket{\psi})\leq\qcorr(\ket{\theta}) = \sum_{i=1}^k \lceil\log_2 r_i\rceil \leq 2\sum_{i=1}^{k-1} n_i \le \Big(2-\frac{2}{k}\Big)\sum_{i=1}^{k} n_i = \Big(2-\frac{2}{k}\Big) \qcorr(\rho).
\end{align*}
\end{proof}

In the above theorem, the left inequality is tight when $\rho$ is a pure state. The following proposition shows that the right inequality is also tight by giving an example of tripartite state $\rho$ with $\qcorr(\rho) = 3$ and $r(\rho) = 4$. Recall that the 3-qubit GHZ state is $\ket{GHZ} = \frac{1}{\sqrt{2}} (\ket{000} + \ket{111})$ and the 3-qubit W state is $\ket{W} = \frac{1}{\sqrt{3}} (\ket{001} + \ket{010} + \ket{100})$.
\begin{Prop}
    For $\rho_0 = \frac{1}{2}\ket{GHZ}\bra{GHZ}+\frac{1}{2}\ket{W}\bra{W}$, we have $\qcorr(\rho_0) = 3$ and $r(\rho_0)=4$.
\end{Prop}
\begin{proof}
Since $\rho_0$ is a 3-qubit state, the three players can simply
share itself as the seed (and then do nothing), so $\qcorr(\rho_0)
\le 3$. We will next show that $r(\rho_0) = 4$, which implies
$\qcorr(\rho_0) \ge 3$ by Theorem \ref{thm:qcc_intro}. Therefore
$\qcorr(\rho_0) = 3$.

We now prove that $r(\rho_0) = 4$. Suppose the three qubits of
$\rho_0$ are possessed by \alice, \bob, and \charlie respectively.
One simple purification is
\begin{align*}
    \ket{\psi_0}=\frac{1}{\sqrt{2}}\ket{GHZ}\ket{1}+\frac{1}{\sqrt{2}}\ket{W}\ket{0},
\end{align*}
where the last qubit is introduced by one player, say, \charlie. Since $\ket{\psi_0}$ has only 4 qubits, $r(\rho_0) \le 4$. We shall prove that $r(\rho_0) \ge 4$.

Suppose the three qubits of $\rho_0$ are possessed by \alice, \bob,
and \charlie respectively. For convenience, we call these three
qubits the main system. Then an arbitrary purification of $\rho_0$
in $\h_{A}\otimes \h_{A_1}\otimes \h_B \otimes \h_{B_1}\otimes \h_C
\otimes \h_{C_1}$ could be expressed as
\begin{align*}
    \ket{\psi}=\frac{1}{\sqrt{2}}\ket{GHZ}\ket{u_0}+\frac{1}{\sqrt{2}}\ket{W}\ket{u_1},
\end{align*}
where $\ket{u_0}$ and $\ket{u_1}$ are orthogonal, and they are
composed by all the ancillary systems introduced by the three
players. Note that it is possible that some of the players do not
have ancillary systems. Without loss of generality, we suppose some
of the qubits in $\ket{u_i}$ belong to \alice. We trace out the two
qubits of \bob and \charlie in the main systems from $\ket{\psi}$,
and get
\begin{align}\label{eq:rhoa}
    \rho_a=&\tr_{\h_{B}\otimes\h_{C}}\ket{\psi}\bra{\psi}\\
    =&\left(\frac{1}{2}\ket{0}\ket{u_0}+\frac{1}{\sqrt{6}}\ket{1}\ket{u_1}\right)\left(\frac{1}{2}\bra{0}\bra{u_0}+\frac{1}{\sqrt{6}}\bra{1}\bra{u_1}\right)\\
    &+\frac{1}{4}\ket{1}\bra{1}\otimes\ket{u_0}\bra{u_0}+\frac{1}{3}\ket{0}\bra{0}\otimes\ket{u_1}\bra{u_1},
\end{align}
where the first qubit belongs to \alice, and the rest is all the ancillary systems combined. Continue to trace out \bob's ancillary system and \charlie's ancillary system, then we obtain \alice's reduced density matrix $\rho'_a$. Similarly, we can define $\rho'_b$ or $\rho'_c$, provided \bob or \charlie has a nontrivial part in $\ket{u_i}$.

We now prove that at least one of $\rho'_a$, $\rho'_b$ and $\rho'_c$ has a rank at least $3$. If this is the case, say $\rank(\rho_a') \ge 3$, then \alice needs at least 2 qubits. Since \bob and \charlie each needs at least 1 qubit, $\qcorr(\ket{\psi})\ge 4$.

If $\ket{u_i}$ is only at \alice's side, \ie, only \alice introduces
an ancillary system, then $\rho'_a=\rho_a$, which has rank $3$. Now
suppose that \bob also introduces an ancillary system. We claim that
if one of $\ket{u_0}$ and $\ket{u_1}$ is not a product state across
$(A,BC)$, then one of $\rho'_a$, $\rho'_b$ and $\rho'_c$ has rank at
least $3$. Indeed, suppose $\ket{u_0}$ is not a product state across
$(A,BC)$, then
$\rank(\tr_{\h_{B_1}\otimes\h_{C_1}}\ket{u_0}\bra{u_0})\geq2$. Note
that the three components in \eqref{eq:rhoa} are orthogonal, thus
$\rank(\rho'_a)\geq\rank(\tr_{\h_{B_1}\otimes\h_{C_1}}\ket{u_0}\bra{u_0})+\rank(\tr_{\h_{B_1}\otimes\h_{C_1}}\ket{u_1}\bra{u_1})$,
which means $\rank(\rho'_a)\geq3$. Therefore, we only need to take
care of the situation where $\ket{u_0}$ and $\ket{u_1}$ are product
states. Since they are orthogonal, without loss of generality we
could express them as $\ket{u_0}=\ket{u_{0,a}}\ket{v_{0,bc}}$ and
$\ket{u_1}=\ket{u_{1,a}}\ket{v_{1,bc}}$, where $\ket{u_{0,a}}$,
$\ket{u_{1,a}}\in \h_{A_1}$, $\ket{v_{0,bc}}$, $\ket{v_{1,bc}}\in
\h_{B_1}\otimes\h_{C_1}$, with either $\qip{u_{0,a}}{u_{1,a}} = 0$
or $\qip{u_{0,bc}}{u_{1,bc}} = 0$. In this way,
\begin{align*}
    \ket{\psi}=\frac{1}{2}(\ket{000}+\ket{111})\ket{u_{0,a}}\ket{v_{0,bc}}+\frac{1}{\sqrt{6}}(\ket{001}+\ket{010}+\ket{001})\ket{u_{1,a}}\ket{v_{1,bc}}.
\end{align*}
It is not difficult to verify that the rank of
$\rho'_{bc}=\tr_{\h_{A}\otimes\h_{A_1}}\ket{\psi}\bra{\psi}$ is at
least $3$. Meanwhile, it holds that
$\rank(\rho'_{bc})=\rank(\rho'_{a})$. Hence,
$\rank(\rho'_{a})\geq3$, and this completes the proof.
\end{proof}

Next we consider the other extreme, when $\rho$ is a multipartite
classical state, \ie, a multipartite probability distribution.
Recall that for a classical distribution $P$ on $\mcX$, we often
identify it with $\rho = \sum_x P(x)\ket{x}\bra{x}$. Also recall
that for a nonnegative tensor $P = [P(x_1,...,x_k)]_{x_1,...,x_k}$,
its PSD-rank $\prank{k}(P)$ is the minimum $r$ \st there are
$r\times r$ PSD matrices $C_{x_1}^{(1)}, ..., C_{x_k}^{(k)}\succeq
0$ with $P(x_1,...,x_k) = \sum_{i,j=1}^r C_{x_1}^{(1)}(i,j) \cdots
C_{x_k}^{(k)}(i,j)$.

\begin{Thm-qccP}[Restated]
Suppose $P = [P(x_1,...,x_k)]_{x_1,...,x_k}$ is a probability distribution on $\mcX_1\times \cdots \times \mcX_k$. Then we have
\begin{align*}
    \frac{k}{2k-2}\big\lceil\log_2\prank{k}(P)\big\rceil \leq \qcorr(P)\leq k\big\lceil\log_2\prank{k}(P)\big\rceil.
\end{align*}
\end{Thm-qccP}
\begin{proof}
We first prove the right inequality. Let $r=\prank{k}{(P)}$, then there exist positive semi-definite matrices $\{C_{x_i}^{(i)}: i\in [k], x_i\in \mcX_i\}$ \st for any $x = (x_1,...,x_k)$, it holds that $P(x) = \sum_{i,j=1}^r \prod_{t=1}^k C_{x_t}^{(t)} (i,j)$. For $i\in[r]$,
let $\ket{u_{x_t}^i}$ be the $i$-th column of $\sqrt{C_{x_t}^{(t)}}$. Then we have that $\qip{u_{x_t}^j}{u_{x_t}^i} = C_{x_t}^{(t)} (i,j)$. We now define a pure state $\ket{\psi} \in \bigotimes_{t=1}^r (\h_{A_t}\otimes\h_{A_t'}\otimes \h_{A_t''})$ as follows.
\begin{align*}
    \ket{\psi} = \sum_{i=1}^r \bigotimes_{t=1}^k \sum_{x_t} (\ket{x_t} \otimes \ket{x_t} \otimes \ket{u_{x_t}^i}).
\end{align*}
For each $t$, tracing out the second and the third registers gives
\begin{align*}
 & \tr_{\h_{A_t'} \otimes \h_{A_t''}} \ket{\psi} \bra{\psi} \\
= & \sum_{x_1,...,x_k} \ket{x_1}\bra{x_1} \otimes \cdots \otimes \ket{x_k}\bra{x_k} \left( \sum_{i,j=1}^r \prod_{t=1}^k \qip{u_{x_t}^j}{u_{x_t}^i} \right)   \\
= & \sum_{x_1,...,x_k} \ket{x_1}\bra{x_1} \otimes \cdots \otimes \ket{x_k}\bra{x_k} \left( \sum_{i,j=1}^r \prod_{t=1}^k  C_{x_t}^{(t)} (i,j) \right)\\
= & \sum_{x_1,...,x_k} P(x_1,...,x_k) \cdot \ket{x_1}\bra{x_1} \otimes \cdots \otimes \ket{x_k}\bra{x_k} .
\end{align*}
Thus $\ket{\psi}$ is actually a purification of $\rho$, and Theorem
\ref{thm:qcc_intro} implies that
$\qcorr(\rho)\leq\qcorr(\ket{\psi})$. Further note that
$\qcorr(\ket{\psi})\leq k\big\lceil\log_2 r\big\rceil$ by Theorem
\ref{thm:qccforpure}. We thus show that $\qcorr(\rho) \leq
k\big\lceil\log_2 r \big\rceil$.

For the left inequality, suppose $\ket{\psi'}$ is a pure state in
$\bigotimes_{i=1}^k (\h_{A_i} \otimes \h_{A_i'})$ that achieves the
optimum of $r(\rho)$ in Theorem \ref{thm:qcc_intro}, then this
theorem tells us that
\begin{align}\label{eq:fordis}
\qcorr(\rho) \geq \frac{k}{2k-2} \qcorr(\ket{\psi'}) = \frac{k}{2k-2} \sum_{i=1}^k \big\lceil\log_2 r_i\big\rceil \geq \frac{k\log_2(\prod_{i=1}^k r_i)}{2k-2},
\end{align}
where $r_i$ is the dimension of the reduced density matrix of $\ket{\psi'}$ on the $i$-th player. According to Lemma \ref{lem:qccforpure}, $\ket{\psi'}$
could be expressed as
\begin{align*}
   \ket{\psi'}=\sum_{i=1}^R a_{i}\ket{\alpha_i^1}\cdots \ket{\alpha_i^k}.
\end{align*}
Here $R = \prod_{j=1}^k r_j$, and for $i\in[R]$, $\ket{\alpha_i^j}\in\h_{A_j} \otimes \h_{A_j'}$. It should be pointed out
that for different $i$ and $i'$, $\ket{\alpha_i^j}$ and
$\ket{\alpha_{i'}^j}$ might be the same. In this way, $\ket{\psi'}$
could also be written as
\begin{align*}
    \ket{\psi'}=\sum_{i=1}^R \bigotimes_{j=1}^k \left(\sum_{x_j} \ket{x_j}\otimes\ket{u_{x_j}^i}\right).
\end{align*}
Recall that $\ket{\psi'}$ is a purification of $\rho$, so
\begin{align*}
\rho = & \tr_{\h_{A_1'} \otimes ... \otimes \h_{A_k'}} \ket{\psi'} \bra{\psi'} \\
 = & \sum_{x_1,...,x_k} \ket{x_1}\bra{x_1} \otimes \cdots \otimes\ket{x_k}\bra{x_k} \left( \sum_{i,i'=1}^R \prod_{j=1}^k\qip{u_{x_j}^{i'}}{u_{x_j}^i} \right)   \\
= & \sum_{x_1,...,x_k} P(x_1,...,x_k) \ket{x_1}\bra{x_1} \otimes \cdots \otimes\ket{x_k}\bra{x_k}  .
\end{align*}
Note that for any $x$, the $R\times R$ matrix $C_{x}$ with
$C_x(j,i)=\qip{u_{x}^j}{u_x^i}$ is positive. So by the definition of
PSD-rank, we have that $\prank{k}(P)\leq R = \prod_{j=1}^k r_j$.
Combining this result with Eq.\eqref{eq:fordis}, we get that
$\qcorr(\rho)\geq\frac{k}{2k-2}\big\lceil\log_2\prank{k}(P)\big\rceil$,
which completes the proof.
\end{proof}

\section{Quantum Communication Complexity of Multipartite States}

In this section, we study communication complexity of generating multipartite states and prove the results in Section \ref{intro:4}.

\begin{Thm-qcomm_pure}[Restated]
Suppose $\ket{\psi}$ is a $k$-partite pure state, and
$\m(\ket{\psi}) = \sum_{j=1}^k \lceil \log_2 \rank(\rho_i) \rceil $ where $\rho_i$ is $\ket{\psi}$ reduced to Player $i$'s space. Then
\begin{align*}
   \frac{1}{2}\m(\ket{\psi})\leq\qcomm(\ket{\psi}) \leq \frac{k-1}{k}\m(\ket{\psi}).
\end{align*}
\end{Thm-qcomm_pure}

\begin{proof}
Let us prove the upper bound first. By Theorem \ref{thm:qccforpure}, we can assume that the players can generate $\rho$ by local operations on the seed state $\sigma$ of size $\m(\ket{\psi})$. Suppose that Player $i$'s part of $\sigma$ has the largest number of qubits, then this player can prepare $\sigma$ and send to other players their parts. The communication cost is thus at most $\frac{k-1}{k} \m(\ket{\psi})$.

For the lower bound, suppose that Player $i$ and Player $j$  communicate $c_{ij}$ qubits in an optimal communication protocol generating $\ket{\psi}$, starting from a product state. Considering the linearity of quantum operations and that the target state is pure, we can assume without loss of generality that the seed state is also pure. Denote by $r_i = \rank(\rho_i)$ where $\rho_i$ is $\ket{\psi}$ reduced to Player $i$'s space. Since exchanging $r$ qubits can only increase the Schmidt rank between Player $i$ and the rest of the players by at most $2^r$, we have that
\begin{align*}
    r_i &\leq 2^{\sum_{j:j\ne i} c_{ij}}.
\end{align*}
Putting communication among all pairs of players together, we have
\begin{align*}
\qcomm(\ket{\psi}) = \sum_{\{i,j\}: i\ne j} c_{ij} = \frac{1}{2} \sum_i \sum_{j: j\ne i} c_{ij}
\ge \frac{1}{2} \sum_i \lceil \log_2 r_i \rceil
\ge \frac{1}{2}\m(\ket{\psi}).
\end{align*}
\end{proof}

Both bounds in the above theorem are tight. For the upper bound,
consider the 3-qubit GHZ state $\ket{\psi} = \frac{1}{\sqrt{2}}
(\ket{000} + \ket{111})$ shared by \alice, \bob and \charlie. It is
not hard to see that $\m(\ket{\psi}) = 3$ and
$\qcomm(\ket{\psi})=2$. For the lower bound, consider an EPR pair
$\ket{\psi} = \frac{1}{\sqrt{2}} (\ket{00} + \ket{11})$ shared by
two players. It has $\m(\ket{\psi}) = 2$ and $\qcomm(\ket{\psi}) =
1$.

\medskip
In Theorem \ref{thm:qcc_intro} and its later comment on tightness of
the bounds we have seen that the correlation complexity of a mixed
quantum state $\rho$ is in general different than that of (even a
best) purification of $\rho$. The next theorem shows that for
communication complexity, generating a mixed quantum state is the
same as generating a purification of it.

\begin{Thm-qcomm_topure}[Restated]
For any $k$-partite quantum state $\rho$,
\begin{align*}
    \qcomm(\rho) = \min \{\qcomm(\ket{\psi}): \ket{\psi} \text{ is a purification of }\rho\}.
\end{align*}
\end{Thm-qcomm_topure}
\begin{proof}
It is clear that for any purification $\ket{\psi}$, $\qcomm(\rho) \le \qcomm(\ket{\psi})$ since one can just generate $\ket{\psi}$ and then trace out some part to get $\rho$.

For the other direction, suppose $r = \qcomm(\rho)$, then starting
from $\otimes_{i=1}^k \ket{0}$, the players can generate $\rho$ by
local operations and communicating $r$ qubits. Here all local
operations can be assumed to be first to append some ancilla and
then perform a unitary operation and finally trace out some parts.
If the players do not trace out any part, then at the end of the
protocol, they would have a pure state as a purification
$\ket{\psi}$ of $\rho$. Thus $\qcomm(\rho) \ge \qcomm(\ket{\psi})$.
\end{proof}

The following result compares $\qcorr(\rho)$ and $\qcomm(\rho)$ for
general multipartite quantum states.
\begin{Thm-comm-vs-corr}[Restated]
For any $k$-partite quantum state $\rho$, it holds that
\begin{align*}
    \frac{k}{k-1}\qcomm(\rho)\leq\qcorr(\rho)\leq2\qcomm(\rho).
\end{align*}
\end{Thm-comm-vs-corr}
\begin{proof}
The left inequality can be easily proved using the same argument as
the lower bound proof of Theorem \ref{thm:qcomm_pure}.

For the right inequality, according to Theorem
\ref{thm:qcomm_topure}, we could find a purification $\ket{\psi}$ of
$\rho$ in $\bigotimes_{i=1}^k (\h_{A_i}\otimes\h_{A_i'})$ such that
$\qcomm(\rho) = \qcomm(\ket{\psi})$. Then Theorem
\ref{thm:qcomm_pure} indicates that
$\qcomm(\ket{\psi})\geq\frac{1}{2}\qcorr(\ket{\psi})$. Combing these
results with Theorem \ref{thm:qcc_intro}, we obtain that
\begin{align*}
\qcorr(\rho)\leq\qcorr(\ket{\psi}) \leq 2\qcomm(\ket{\psi})=2\qcomm(\rho).
\end{align*}
\end{proof}

\section{Approximate Quantum Correlation Complexity of Bipartite States}
In this section, we study the correlation complexity of generating
bipartite states approximately, and prove the results mentioned in
Section \ref{intro:1}. We will first consider two extreme cases:
quantum pure states and classical distributions, and then general
quantum mixed states.

\subsection*{Quantum pure states.}
For quantum pure states, we will first show that the following two
approximations are equivalent. Recall that for a state $\rho\in
\h_A\otimes \h_B$,
\[
    \qcorr_\epsilon(\rho) = \min\{\qcorr(\rho'): \rho' \in  \h_{AB}, \F(\rho, \rho') \geq 1 - \epsilon\},
\]
and
\begin{align*}
    \qcorr'_\epsilon(\rho) & = \min\big\{ \qcorr_\epsilon^{pure}(\ket{\varphi}): \ket{\varphi}\in  \h_{A_1ABB_1},  \rho= \tr_{\h_{A_1}\otimes\h_{B_1}} \ket{\varphi}\bra{\varphi} \big\} \\
    & = \min\big\{ \big\lceil\log_2\srank_\epsilon(\ket{\varphi})\big\rceil: \ket{\varphi}\in  \h_{A_1ABB_1},  \rho= \tr_{\h_{A_1}\otimes\h_{B_1}} \ket{\varphi}\bra{\varphi} \big\}.
\end{align*}
We will need a result in \cite{JSWZ13} which says that pure states
can be optimally approximated by given other \emph{pure} state.
\begin{Lem}[\cite{JSWZ13}]\label{lem:approx-pure}
    For a bipartite pure state $\ket{\psi}$ with Schmidt coefficients $\lambda_1 \ge ... \ge \lambda_N$, $\qcorr_\epsilon(\ket{\psi}) = \qcorr_\epsilon^{pure}(\ket{\psi}) = \lceil \log_2 r \rceil$, where $r$ is the minimum integer \st $\sum_{i=1}^r \lambda_i^2 \ge (1-\epsilon)^2$.
\end{Lem}

\begin{Prop}\label{lem:qcc1}%{\em (restatement of Lemma \ref{thm:qcc1_intro})}
For any quantum state $\rho$ in $\h_A \otimes \h_{B}$,
$\qcorr_\epsilon(\rho)=\qcorr'_\epsilon(\rho)$.
\end{Prop}
\begin{proof}
$\qcorr_\epsilon(\rho) \ge \qcorr'_\epsilon(\rho)$: Suppose that
$\rho' \in  \h_A \otimes \h_B$, $\F(\rho, \rho') \geq 1 - \epsilon$
and $\qcorr_\epsilon(\rho) = \qcorr(\rho')$. By Lemma 2.2 of
\cite{JSWZ13}, there is a purification $\ket{\psi}$ in $A_1ABB_1$ of
$\rho'$ \st $\qcorr(\rho') =
\big\lceil\log_2\srank(\ket{\psi})\big\rceil$. By Uhlmann's theorem,
there exists a purification of $\rho$ in $A_1ABB_1$, say
$\ket{\alpha}$, and
$\F(\ket{\alpha}\bra{\alpha},\ket{\psi}\bra{\psi}) = \F(\rho, \rho')
\geq 1-\epsilon$. (We assume that the $\ket{\alpha}$ and
$\ket{\psi}$ are in the same extended space $\h_{A_1ABB_1}$ since
otherwise we can use the union of the two extended spaces.) Thus
\[\qcorr'_\epsilon(\rho) \le \big\lceil\log_2\srank_\epsilon(\ket{\alpha})\big\rceil \le \big\lceil\log_2\srank(\ket{\psi})\big\rceil = \qcorr_\epsilon(\rho).\]

$\qcorr_\epsilon(\rho) \le \qcorr'_\epsilon(\rho)$: Suppose
$\qcorr'_\epsilon(\rho) =
\big\lceil\log_2\srank_\epsilon(\ket{\varphi})\big\rceil$, and
$\rho= \tr_{\h_{A_1}\otimes\h_{B_1}} \ket{\varphi}\bra{\varphi}$.
Then one can find another pure state $\ket{\beta}$ in $A_1ABB_1$,
such that $\big\lceil\log_2\srank(\ket{\beta})\big\rceil =
\big\lceil\log_2\srank_\epsilon(\ket{\varphi})\big\rceil =
\qcorr'_\epsilon(\rho) $, and
$\F(\ket{\beta}\bra{\beta},\ket{\varphi}\bra{\varphi})\geq1-\epsilon$.
Since partial trace does not decrease the fidelity \cite{NC00}, we
know that $\F(\tr_{\h_{A_1}\otimes\h_{B_1}}\ket{\beta}\bra{\beta},
\rho)\geq 1-\epsilon$. By the definition of $\qcorr_\epsilon(\rho)$,
it holds that $\qcorr'_\epsilon(\rho)\geq\qcorr_\epsilon(\rho)$,
which completes the proof.
\end{proof}

\subsection*{Classical distributions.} Next we consider to approximate classical distributions.
Recall that Lemma \ref{lem:approx-pure} implies that the most
efficient approximate generation of a pure state can be achieved by
another pure state. In the same spirit, the following theorem shows
that the most efficient approximate generation of a classical state
can be achieved by another classical state, and the correlation
complexity is completely determined by the approximate PSD-rank.

\begin{Thm-aqcorr_P}[Restated]
For any classical state $P = [P(x,y)]_{x,y}$, \[\qcorr_\epsilon(P) =
\qcorr_\epsilon^{\cla}(P) = \lceil \log_2 \apprank{\epsilon}(P)
\rceil.\]
\end{Thm-aqcorr_P}

\begin{proof}
     For the first equality, we only need to prove that $\qcorr_\epsilon(P) \ge \qcorr_\epsilon^{\cla}(P)$ (since the other direction holds by definition). Given an approximation $\rho'$ to $P$ with $F(P,\rho') \ge 1-\epsilon$ and $\size(\rho') = \qcorr_\epsilon(P)$, we measure $\rho'$ in the computational basis of $\h_A\otimes \h_B$ and get a probability distribution $P'$. Note that the same measurement does not change $P$. Since no operation can decrease the fidelity of two states, we have $F(P,P') \ge F(P,\rho') \ge 1-\epsilon$.

    The second equality is immediate from their definitions.
\end{proof}

\subsection*{General quantum mixed states.} We now turn to the case of general bipartite $\sigma$. By combining
Theorem 1.2 of \cite{JSWZ13} and Proposition \ref{lem:qcc1}, we have
the following characterization of $\qcorr_{\epsilon}(\sigma)$.

\begin{Thm-aQCorrgeneral}[Restated]
Let $\sigma$ be an arbitrary quantum state in $\h_A \otimes \h_B$,
and $0<\epsilon<1$. Then $\qcorr_{\epsilon}(\sigma) = \lceil \log_2
r \rceil $, where $r$ is the minimum integer \st there exist a
collection of matrices, $\{A_x\}$'s and $\{B_y\}$'s of the same
column number $l\ge r$, satisfying the following conditions.

\begin{enumerate}
    \item The matrices relate to $\sigma$ by the following equation.
    \begin{align}\label{eq:requireforgeneral}
        \sigma=\sum_{x,x'; y,y'} \ket{x}\bra{x'} \otimes \ket{y}\bra{y'} \cdot \tr \Big( (A_{x'}^\dag A_x)^T (B_{y'}^\dag B_y) \Big)\Big).
    \end{align}

    \item Denoting the $i$-th column of any matrix $M$ by $\ket{M(i)}$, then
    \begin{equation}\label{eq:requireforgeneral2}
        \sum_x\qip{A_x(i)}{A_x(j)}=\sum_y\qip{B_y(i)}{B_y(j)}=0,
    \end{equation}

    \item
    \begin{align}\label{eq:requireforgeneral3}
        \sum_{i=1}^r\Big(\sum_x\qip{A_x(i)}{A_x(i)}\Big)\Big(\sum_y\qip{B_y(i)}{B_y(i)}\Big)\geq1-\epsilon,
    \end{align}
\end{enumerate}
\end{Thm-aQCorrgeneral}
\begin{proof}
Suppose $\qcorr_{\epsilon}(\sigma) = \lceil \log_2
\srank_\epsilon(\ket{\psi}) \rceil$ where $\ket{\psi}$ is a
purification of $\sigma$, given by Proposition \ref{lem:qcc1}. Put
$t = \srank_{\epsilon}(\ket{\psi})$. Suppose the Schmidt
decomposition of $\ket{\psi}$ is
\begin{equation}\label{eq:forgeneral2}
\ket{\psi} = \sum_{i=1}^s \Big(\sum_{x} \ket{x} \otimes
\ket{v_x^i}\Big) \otimes \Big(\sum_{y} \ket{y} \otimes
\ket{w_y^i}\Big),
\end{equation}
thus the Schmidt coefficients are
$a_i=\sum_x\qip{v_x^{i}}{v_x^{i}}\sum_y\qip{w_y^{i}}{w_y^{i}}$,
$1\leq i\leq s$. For each $x$, set matrices $A_x \defeq
(\ket{v_x^1}, \ket{v_x^2}, \ldots, \ket{v_x^s})$. Similarly, for
each $y$ set matrices $B_y \defeq (\ket{w_y^1}, \ket{w_y^2}, \ldots,
\ket{w_y^s})$. Then it can be verified that
Eq.\eqref{eq:requireforgeneral} holds. In addition, the
orthogonality of $\sum_{x} \ket{x} \otimes \ket{v_x^i}$ (and that of
$\sum_{y} \ket{y} \otimes \ket{w_y^i}$) for different $i$'s
translates to Eq.\eqref{eq:requireforgeneral2}.

Without loss of generality, we assume that the coefficients
$a_{i}$'s are in the decreasing order. By Lemma 5.1 of
\cite{JSWZ13}, we have that
\begin{equation*}
\sum_{i=1}^t\Big(\sum_x\qip{A_x(i)}{A_x(i)}\Big)\Big(\sum_y\qip{B_y(i)}{B_y(i)}\Big)
= \sum_{i=1}^ta_i \geq 1-\epsilon.
\end{equation*}
Therefore all three conditions hold for $\{A_x\}$ and $\{B_y\}$,
implying that $r\leq t$, and that $\qcorr_{\epsilon}(\sigma)\geq
\lceil \log_2 r \rceil$.

For the other direction, given that matrices $\{A_x\}$'s and
$\{B_y\}$'s satisfy the requirements, it can be verified that
\begin{equation}\label{eq:forgeneral3}
\ket{\tilde{\psi}} = \sum_{i=1}^l \Big(\sum_{x} \ket{x} \otimes
\ket{A_x(i)}\Big) \otimes \Big(\sum_{y} \ket{y} \otimes
\ket{B_y(i)}\Big)
\end{equation}
is a purification of $\sigma$ in $\h_{A_1}\otimes \h_A\otimes \h_B
\otimes \h_{B_1}$. Again assuming the decreasing order of the
Schmidt coefficients and taking the $r$ leading terms,
\begin{equation}\label{eq:forgeneral4}
\ket{\tilde{\psi'}} = \sum_{i=1}^r \Big(\sum_{x} \ket{x} \otimes
\ket{A_x(i)}\Big) \otimes \Big(\sum_{y} \ket{y} \otimes
\ket{B_y(i)}\Big),
\end{equation}
then Eq.\eqref{eq:requireforgeneral3} means that
$|\qip{\tilde{\psi}}{\tilde{\psi'}}|\geq 1-\epsilon$. Since
$\srank(\ket{\tilde{\psi'}})\leq r$, it holds that
$\srank_{\epsilon}(\ket{\tilde{\psi}})\leq r$, and according to
Proposition \ref{lem:qcc1} we know that
$\qcorr_{\epsilon}(\sigma)\leq \lceil \log_2 r \rceil$, which
completes the proof.
\end{proof}

\section{Approximate Quantum Correlation Complexity of Multipartite Pure States}

In this section, we consider approximation in generating multipartite pure states, and prove the results in Section \ref{intro:5}. Recall that for a $k$-partite pure state $\ket{\psi}$, $r_i^{(\epsilon)} = \srank^{(A_i,A_{-i})}_{\epsilon}(\ket{\psi})$.

\begin{Thm-aQccforPure}[Restated]
Let $\ket{\psi}\in \bigotimes_{i=1}^k \h_{i}$ be a $k$-partite
state, $\epsilon>0$, and
%\begin{align*}
    $\m_{\epsilon}(\ket{\psi}) = \sum_{i=1}^k \big\lceil\log_2 r_{i}^{(\epsilon)}\big\rceil$.
%\end{align*}
Then
\begin{align*}
\m_{\epsilon}(\ket{\psi})\leq\qcorr_{\epsilon}^{pure}(\ket{\psi})\leq
\m_{\epsilon/k}(\ket{\psi}).
\end{align*}
\end{Thm-aQccforPure}
\begin{proof}
Lower bound: Suppose $\ket{\phi}$ is a pure state in
$\bigotimes_{i=1}^k \h_{A_i}$ such that
$\qcorr_\epsilon^{pure}(\ket{\psi}) = \qcorr(\ket{\phi})$ and
$\F(\ket{\psi}\bra{\psi}, \ket{\phi}\bra{\phi}) \geq 1 - \epsilon$.
Suppose $\sigma_i$ is the reduced density matrix of $\ket{\phi}$ in
Player $i$'s system, and its rank is $s_i$, which is also
$\srank^{(A_i,A_{-i})}_{\epsilon}(\ket{\psi})$. Then it holds that
$s_i\geq r_i^{(\epsilon)}$. According to Theorem
\ref{thm:qccforpure}, we have that
\begin{align*}
\qcorr_\epsilon^{pure}(\ket{\psi}) = \qcorr(\ket{\phi}) = \sum_{i=1}^K \big\lceil\log_2 s_i\big\rceil \geq
\m_\epsilon(\ket{\psi}).
\end{align*}

Upper bound: Lemma \ref{lem:qccforpure} shows that
$\ket{\psi}$ could be written as
\begin{align*}
  \ket{\psi} = \sum_{i_j \le r_j} a_{i_1 ... i_k} \ket{\lambda_{i_1}}\cdots \ket{\lambda_{i_k}}
\end{align*}
where $\ket{\lambda_{i_j}}$ is the $i_j$-th eigenvector of $\rho_j$.
Arrange $i_j$ in decreasing order of the eigenvalues of $\rho_j$.
Since $\ket{\psi}$ is a pure state,
\begin{align*}
   \sum_{i_1,...,i_k}|a_{i_1...i_k}|^2 = 1.
\end{align*}

According to the definition of $r_j^{(\epsilon)}$, we have that
\begin{align*}
   \sum_{i_j=r_j^{(\epsilon/k)}+1}^{r_j}\sum_{i_{-j}}|a_{i_1...i_k}|^2 \leq \epsilon/k,
\end{align*}
where we have used Lemma 5.1 of \cite{JSWZ13} and the fact that
\begin{align*}
   \bra{\lambda_{i_j}}\rho_j\ket{\lambda_{i_j}}=\sum_{i_{-j}}|a_{i_1...i_k}|^2.
\end{align*}
Thus,
\begin{align*}
\sum_{i_1=1}^{r_1^{(\epsilon/k)}} \cdots
\sum_{i_k=1}^{r_k^{(\epsilon/k)}} |a_{i_1...i_k}|^2 \geq & 1 -
\sum_{j=1}^k \sum_{i_j=r_j^{(\epsilon/k)}+1}^{r_j}\sum_{i_{-j}}
|a_{i_1...i_k}|^2
 \geq 1 - \epsilon.
\end{align*}
We now consider a pure state defined as
\begin{align*}
   \ket{\phi'}=\frac{1}{\sqrt{m}}\sum_{i_1=1}^{r_1^{(\epsilon/k)}} \cdots \sum_{i_k=1}^{r_k^{(\epsilon/k)}}  a_{i_1...i_k} \ket{\lambda_{i_1}}\cdots \ket{\lambda_{i_k}},
\end{align*}
where $m = \sum_{i_1=1}^{r_1^{(\epsilon/k)}} \cdots
\sum_{i_k=1}^{r_k^{(\epsilon/k)}} |a_{i_1...i_k}|^2$. It is not
difficult to prove that $\F(\ket{\psi}\bra{\psi},
\ket{\phi'}\bra{\phi'}) \geq \sqrt{1 - 2\epsilon}\approx1-\epsilon$.
Moreover, according to Theorem \ref{thm:qccforpure}, it holds that
$\qcorr(\ket{\phi'})\leq \m_{\epsilon/k}(\ket{\psi})$. According to
the definition of $\qcorr_{\epsilon}^{pure}(\ket{\psi})$, we obtain
that $\qcorr_{\epsilon}^{pure}(\ket{\psi})\leq
\m_{\epsilon/k}(\ket{\psi})$.
\end{proof}

From the proof we can see that the upper bound can be generalized to the following.

\begin{Thm}
Suppose
\begin{align*}
R = \min_{S_1,\ldots ,S_k}\{\prod_{i=1}^k |S_i|: S_i\subset [r_i], \sum_{i_1\in S_1, \ldots, i_k \in S_k} |a_{i_1...i_k}|^2 \geq 1 - \varepsilon\}.
\end{align*}
Then $\qcorr_{\epsilon}^{pure}(\ket{\psi})\leq \log_2{\lceil R\rceil}$.
\end{Thm}

Finally, we consider the relationship between $\qcorr_{\epsilon}(\ket{\psi})$ and $\qcorr^{pure}_{\epsilon}(\ket{\psi})$.
\begin{Thm-PureandMixed}[Restated]
Let $\ket{\psi}\in \h_{A_1}\otimes\cdots\otimes \h_{A_k}$ be a pure state and $\epsilon >0$. Then
\begin{align*}
\frac{k}{2k-2}\qcorr^{pure}_{k\epsilon}(\ket{\psi})\leq\qcorr_{\epsilon}(\ket{\psi})\leq
\qcorr^{pure}_{\epsilon}(\ket{\psi}).
\end{align*}
\end{Thm-PureandMixed}
\begin{proof}
The second inequality $\qcorr_{\epsilon}(\ket{\psi})\leq
\qcorr^{pure}_{\epsilon}(\ket{\psi})$ holds by definition. For the
first inequality, according to the definition of
$\qcorr_{\epsilon}(\ket{\psi})$, there exists a $\rho \in
\h_{A_1}\otimes\cdots\otimes \h_{A_k}$ such that
$\qcorr_{\epsilon}(\ket{\psi})=\qcorr(\rho)$ and
$\F(\ket{\psi}\bra{\psi},\rho)\geq1-\epsilon$. Then Theorem
\ref{thm:qcc_intro} implies that there exists a purification
$\ket{\phi}\in \bigotimes_{i=1}^k (\h_{A_i}\otimes\h_{A_i'})$ of
$\rho$ such that \[\qcorr(\rho)\geq
\frac{k}{2k-2}\qcorr(\ket{\phi}).\] Thus, we could find a pure state
$\ket{\theta}\in\bigotimes_{i=1}^k \h_{A_i'}$ that makes
$\F(\ket{\phi}\bra{\phi},\ket{\phi'}\bra{\phi'})\geq1-\epsilon$,
where $\ket{\phi'}=\ket{\psi}\otimes\ket{\theta}$. By the definition
of $\qcorr^{pure}_{\epsilon}(\ket{\phi'})$, we have that
\[\qcorr(\ket{\phi})\geq\qcorr^{pure}_{\epsilon}(\ket{\phi'}).\]
Combining the above two inequalities, we see that
\[\qcorr(\rho)\geq\frac{k}{2k-2}\qcorr^{pure}_{\epsilon}(\ket{\phi'})\geq\frac{k}{2k-2}\m_\epsilon(\ket{\phi'}),\]
where the last inequality comes from Theorem
\ref{thm:aQccforPure_intro}. According to Lemma 5.2 of
\cite{JSWZ13}, we have that $\m_\epsilon(\ket{\phi'})\geq
\m_\epsilon(\ket{\psi})$. Applying Theorem
\ref{thm:aQccforPure_intro} again, we eventually get that
$\qcorr^{pure}_{\epsilon}(\ket{\phi'})\geq
\m_\epsilon(\ket{\psi})\geq\qcorr^{pure}_{k\epsilon}(\ket{\psi})$.
This means that
\begin{align*}
\qcorr_{\epsilon}(\ket{\psi})=\qcorr(\rho)\geq\frac{k}{2k-2}\qcorr^{pure}_{\epsilon}(\ket{\phi'})\geq\frac{k}{2k-2}\qcorr^{pure}_{k\epsilon}(\ket{\psi}),
\end{align*}
and the proof is completed.
\end{proof}

\subsection*{Acknowledgments}
This work is supported by the Singapore Ministry of Education Tier 3 Grant and the Core Grants of the Center for Quantum Technologies (CQT), Singapore. Z.W. is also supported in part by the Singapore National Research Foundation under NRF RF Award No. NRF-NRFF2013-13. S.Z. was supported by Research Grants Council of the Hong Kong S.A.R. (Project no. CUHK419011, CUHK419413), and part of the research was done when S.Z. visited CQT and Tsinghua University, latter partially supported by China Basic Research Grant 2011CBA00300 (sub-project 2011CBA00301).

\end{document}